\makeatletter
\def\input@path{{styles/}{styles/}}
\makeatother

\ifx\SoCG\undefined%
    \ifx\CGT\undefined
        \documentclass[12pt]{article}%
        \providecommand{\SoCGVer}[1]{}%
        \providecommand{\CGTVer}[1]{}%
        \providecommand{\RegVer}[1]{#1}%
        \def\UseBibLatex{1}
    \else
        \documentclass[a4paper,USenglish,cleveref, %
                       autoref, thm-restate]{compgeomtop}
        \providecommand{\SoCGVer}[1]{}%
        \providecommand{\CGTVer}[1]{#1}%
        \providecommand{\RegVer}[1]{}%
    \fi
\else%
    \documentclass[a4paper,USenglish,cleveref,autoref,thm-restate]%
    {socg-lipics-v2019}
    \providecommand{\SoCGVer}[1]{#1}%
    \providecommand{\CGTVer}[1]{}%
    \providecommand{\RegVer}[1]{}%
\fi

\RegVer{%
   \usepackage[cm]{fullpage}%
}%

\ifx\sarielComp\undefined%
\newcommand{\SarielComp}[1]{}
\newcommand{\NotSarielComp}[1]{#1}%
\else
\newcommand{\SarielComp}[1]{#1}%
\newcommand{\NotSarielComp}[1]{}%
\fi
\newcommand{\IfPrinterVer}[2]{#2}%

\providecommand{\BibLatexMode}[1]{}
\providecommand{\BibTexMode}[1]{#1}

\ifx\UseBibLatex\undefined%
  \renewcommand{\BibLatexMode}[1]{}
  \renewcommand{\BibTexMode}[1]{#1}
\else
  \renewcommand{\BibLatexMode}[1]{#1}
  \renewcommand{\BibTexMode}[1]{}
\fi

\BibLatexMode{%
   \usepackage[style=alphabetic,backend=biber,bibencoding=utf8]{biblatex}%
   \usepackage{sariel_biblatex}%
}

\usepackage{amsmath}%
\usepackage{amssymb}%
\usepackage{xcolor}%
\usepackage{caption}%
\usepackage{scalerel}%

\SarielComp{\usepackage{sariel_colors}}%

\RegVer{%
   \usepackage[amsmath,thmmarks]{ntheorem}%
   \theoremseparator{.}%
}

\usepackage{titlesec}%
\titlelabel{\thetitle. }%
\usepackage{xcolor}%
\usepackage{mleftright}%
\usepackage{xspace}%
\usepackage{hyperref}%
\usepackage{mathtools}%
\usepackage{stmaryrd}%
\usepackage{marvosym}%

\usepackage{tabularx}%
\usepackage{array}

\RegVer{%
\IfPrinterVer{%
   \usepackage{hyperref}%
}{%
   \usepackage{hyperref}%
   \hypersetup{%
      breaklinks,%
      ocgcolorlinks, colorlinks=true,%
      urlcolor=[rgb]{0.25,0.0,0.0},%
      linkcolor=[rgb]{0.5,0.0,0.0},%
      citecolor=[rgb]{0,0.2,0.445},%
      filecolor=[rgb]{0,0,0.4},
      anchorcolor=[rgb]={0.0,0.1,0.2}%
   }
}%
}

\SoCGVer{%
   \theoremstyle{plain}%

   \newtheorem{remark:unnumbered}[FakeCounter]{Remark}%

   \newtheorem{conjecture}[theorem]{Conjecture}
   \newtheorem{fact}[theorem]{Fact}
   \newtheorem{observation}[theorem]{Observation}
   \newtheorem{invariant}[theorem]{Invariant}
   \newtheorem{question}[theorem]{Question}
   \newtheorem{prop}[theorem]{Proposition}
   \newtheorem{openproblem}[theorem]{Open Problem}

   \theoremstyle{plain}%
   \newtheorem{defn}[theorem]{Definition}
   \newtheorem{problem}[theorem]{Problem}
   \newtheorem{xca}[theorem]{Exercise}
   \newtheorem{exercise_h}[theorem]{Exercise}
   \newtheorem{assumption}[theorem]{Assumption}%

   \newtheorem{proofof}{Proof of\!}%
}%

\RegVer{%
   \theoremseparator{.}%

\theoremstyle{plain}%
\newtheorem{theorem}{Theorem}[section]

\newtheorem{lemma}[theorem]{Lemma}

\theoremstyle{plain}%
\theoremheaderfont{\sf} \theorembodyfont{\upshape}%
\newtheorem*{remark:unnumbered}[theorem]{Remark}%
\newtheorem{remark}[theorem]{Remark}%
\newtheorem{definition}[theorem]{Definition}

\newcommand{\myqedsymbol}{\rule{2mm}{2mm}}

\theoremheaderfont{\em}%
\theorembodyfont{\upshape}%
\theoremstyle{nonumberplain}%
\theoremseparator{}%
\theoremsymbol{\myqedsymbol}%
\newtheorem{proof}{Proof:}%

}

\newcommand{\atgen}{\symbol{'100}}
\newcommand{\SarielThanks}[1]{\thanks{Department of Computer Science;
      University of Illinois; 201 N. Goodwin Avenue; Urbana, IL,
      61801, USA; {\tt sariel\atgen{}illinois.edu}; {\tt
         \url{http://sarielhp.org/}.} #1}}

\newcommand{\TimThanks}[1]{\thanks{Department of Computer Science;
      University of Illinois; 201 N. Goodwin Avenue; Urbana, IL,
      61801, USA; {\tt tzhou28\atgen{}illinois.edu}.
      #1}}

\newcommand{\HLink}[2]{\hyperref[#2]{#1~\ref*{#2}}}
\newcommand{\HLinkSuffix}[3]{\hyperref[#2]{#1\ref*{#2}{#3}}}

\newcommand{\figlab}[1]{\label{fig:#1}}
\newcommand{\figref}[1]{\HLink{Figure}{fig:#1}}

\newcommand{\thmlab}[1]{{\label{theo:#1}}}

\providecommand{\deflab}[1]{\label{def:#1}}
\newcommand{\defref}[1]{\HLink{Definition}{def:#1}}

\newcommand{\seclab}[1]{\label{sec:#1}}
\newcommand{\secref}[1]{\HLink{Section}{sec:#1}}

\newcommand{\lemlab}[1]{\label{lemma:#1}}
\newcommand{\lemref}[1]{\HLink{Lemma}{lemma:#1}}%

\providecommand{\eqlab}[1]{}%
\renewcommand{\eqlab}[1]{\label{equation:#1}}

\definecolor{blue25emph}{rgb}{0, 0, 11}
\providecommand{\emphic}[2]{%
   \textcolor{blue25emph}{%
      \textbf{\emph{#1}}}%
   \index{#2}}

\providecommand{\emphi}[1]{\emphic{#1}{#1}}

\numberwithin{figure}{section}%
\numberwithin{table}{section}%
\numberwithin{equation}{section}%

\providecommand{\remove}[1]{}%
\newcommand{\Set}[2]{\left\{ #1 \;\middle\vert\; #2 \right\}}
\newcommand{\pth}[2][\!]{\mleft({#2}\mright)}%

\newcommand{\ceil}[1]{\left\lceil {#1} \right\rceil}
\newcommand{\floor}[1]{\left\lfloor {#1} \right\rfloor}

\newcommand{\cardin}[1]{\left| {#1} \right|}%

\newcommand{\rScale}{\mathsf{s_r}}%
\newcommand{\rShift}{\mathsf{u_r}}
\newcommand{\canonX}[1]{\text{\sc{canon}}\pth{#1}}%

\renewcommand{\th}{th\xspace}

\renewcommand{\Re}{\mathbb{R}}%

\usepackage[inline]{enumitem}

\newlist{compactenumA}{enumerate}{5}%
\setlist[compactenumA]{topsep=0pt,itemsep=-1ex,leftmargin=1cm,%
   partopsep=1ex,parsep=1.5ex,%
   label=(\Alph*)}%

\newlist{compactenuma}{enumerate}{5}%
\setlist[compactenuma]{topsep=0pt,itemsep=-1ex,partopsep=1ex,parsep=1ex,%
   label=(\alph*)}%

\newlist{compactenumI}{enumerate}{5}%
\setlist[compactenumI]{topsep=0pt,itemsep=-1ex,leftmargin=1cm,partopsep=1ex,parsep=1ex,%
   label=(\Roman*)}%

\newlist{compactenumi}{enumerate}{5}%
\setlist[compactenumi]{topsep=0pt,itemsep=-1ex,partopsep=1ex,parsep=1ex,%
   label=(\roman*)}%

\newlist{compactitem}{itemize}{5}%
\setlist[compactitem]{topsep=0pt,itemsep=-1ex,partopsep=1ex,parsep=1ex,%
   label=\ensuremath{\bullet}}%

\providecommand{\Mh}[1]{#1}%

\newcommand{\SaveContent}[2]{%
   \expandafter\newcommand{#1}{#2}%
}

\newcommand{\longSY}[2]{\SS_{\geq #1}\pth{#2}}%
\newcommand{\shortSY}[2]{\SS_{<#1}\pth{#2}}%

\newcommand{\CongChar}{\Mh{\mathcal{C}}}%
\newcommand{\congChar}{\Mh{\mathcalb{c}}}%
\newcommand{\SETH}{\Term{SETH}\xspace}%
\newcommand{\ThreeSUM}{\Term{3SUM}\xspace}%

\newcommand{\lenTShortX}[1]{\Mh{\Delta}_{#1}}%
\newcommand{\clX}[1]{\Mh{L}\pth{#1}}
\newcommand{\clSX}[1]{\Mh{L_{\mathrm{short}}}\pth{#1}}
\newcommand{\clLX}[1]{\Mh{L_{\mathrm{long}}}\pth{#1}}
\newcommand{\clALX}[1]{\Mh{L_{\mathrm{LONG}}}\pth{#1}}

\newcommand{\ACongOnY}[2]{\mathcal{A}_{#2}^{}\pth{#1}}%
\newcommand{\CongOnY}[2]{\CongChar_{#2}^{}\pth{#1}}%
\newcommand{\congOnY}[2]{\congChar_{#2}^{}\pth{#1}}%
\newcommand{\congX}[1]{\congChar\pth{#1}}%

\newcommand{\Longc}{\CongChar^{}_{\!\ts\geq \tr}}%

\newcommand{\longc}{\congChar^{}_{\!\ts\geq \tr}}%
\newcommand{\longcY}[2]{\congChar^{}_{\!\ts\geq #1}\pth{#2}}%

\newcommand{\QT}{\Mh{\mathcal{T}}}%
\newcommand{\QTM}{\Mh{\mathcal{T}^+}}%

\newcommand{\GSquares}{\Mh{{\mathcal{G}}}}
\newcommand{\GSS}{\GSquares_\SS}

\newcommand{\shortc}{\congChar_{<\tr}}%
\newcommand{\shortcY}[2]{\congChar_{<#1}\pth{#2}}%

\newcommand{\radiusX}[1]{\Mh{\mathsf{r}}\pth{#1}}%

\newcommand{\ts}{\hspace{0.6pt}}

\DeclareFontFamily{U}{BOONDOX-calo}{\skewchar\font=45 }
\DeclareFontShape{U}{BOONDOX-calo}{m}{n}{
  <-> s*[1.05] BOONDOX-r-calo}{}
\DeclareFontShape{U}{BOONDOX-calo}{b}{n}{
  <-> s*[1.05] BOONDOX-b-calo}{}
\DeclareMathAlphabet{\mathcalb}{U}{BOONDOX-calo}{m}{n}
\SetMathAlphabet{\mathcalb}{bold}{U}{BOONDOX-calo}{b}{n}
\DeclareMathAlphabet{\mathbcalb}{U}{BOONDOX-calo}{b}{n}

\newcommand{\tr}{\Mh{\alpha}}%

\newcommand{\myparagraph}[1]{%
   \RegVer{\paragraph*{#1}}%
   \CGTVer{\subparagraph*{#1}}%
}

\newlength{\CharHeight}
\AtBeginDocument{\setlength{\CharHeight}{\fontcharht\font`X}}   %

\IfFileExists{figs/duck.pdf}%
{%
   \renewcommand{\myqedsymbol}{%
      \reflectbox{\includegraphics[height=1.4\CharHeight]%
         {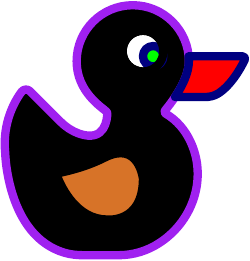}}%
   }%
}{}

\newcommand{\PS}{\Mh{P}}%
\newcommand{\SSet}{\Mh{\Xi}}%
\definecolor{almostblack}{rgb}{0, 0, 0.3}
\newcommand{\emphw}[1]{{\textcolor{almostblack}{\emph{#1}}}}%

\newcommand{\pp}{\Mh{p}}%
\newcommand{\maxSz}{\Mh{\rho}}%
\newcommand{\maxSzX}[1]{\maxSz\pth{#1}}%

\newcommand{\Term}[1]{\textsf{#1}}
\newcommand{\etal}{\textit{et~al.}\xspace}%

\newcommand{\eps}{{\varepsilon}}%
\newcommand{\epsA}{{\xi}}%

\newcommand{\prob}{\Mh{\varsigma}}%

\newcommand{\Frechet}{Fr\'{e}c{h}e{}t\xspace}%

\newcommand{\tldO}{\scalerel*{\widetilde{O}}{j^2}}%

\newcommand{\lenX}[1]{\left\|#1\right\|}

\newcommand{\seg}{\Mh{s}}

\providecommand{\SS}{}%
\renewcommand{\SS}{\Mh{\mathcal{S}}}%

\newcommand{\Grid}{\Mh{\mathsf{G}}\index{grid}}
\newcommand{\sq}{\Mh{\square}}%
\newcommand{\sqY}[2]{\Mh{\square}\pth{#1, #2}}%

\newcommand{\num}{\Mh{\tau}}%

\newcommand{\GridSetY}[2]{\Mh{\mathsf{G}}\pth{#1,#2}}%

\newcommand{\naive}{n{}a{\"\i}v{}e\xspace}

\newcommand{\Sample}{\Mh{R}}%

\newcommand{\Cell}{\Mh{C}}%
\newcommand{\LSet}{\Mh{L}}%
\newcommand{\ancX}[1]{\mathrm{anc}\pth{#1}}
\newcommand{\descX}[1]{\mathrm{desc}\pth{#1}}

\newcommand{\SC}{B}%
\newcommand{\SA}{\mathcal{S}}%
\newcommand{\BadProb}{\varphi}%
\newcommand{\cE}{c_2}%

\BibLatexMode{\bibliography{packedness}}%

\title{How Packed Is It, Really?}

\SoCGVer{%
   \author{Sariel Har-Peled}%
   {Department of Computer Science, University of Illinois, 201
      N. Goodwin Avenue, Urbana, IL 61801, USA}%
   {sariel@illinois.edu}%
   {https://orcid.org/0000-0003-2638-9635}%
   {Work on this paper was partially supported by a NSF AF award
      CCF-1907400.}%
   \author{Timothy Zhou}%
   {Department of Computer Science, University of Illinois, 201
      N. Goodwin Avenue, Urbana, IL 61801, USA}%
   {z.timothy96@gmail.com}%
   {https://orcid.org/0009-0009-1831-4634}%
   {Work on this paper was partially supported by a NSF AF award
      CCF-1907400.}%
}%

\CGTVer{%
   \author{Sariel Har-Peled}%
   {Department of Computer Science, University of Illinois, 201
      N. Goodwin Avenue, Urbana, IL 61801, USA. \and
      \url{https://sarielhp.org}.|}%
   {sariel@illinois.edu}
   {https://orcid.org/0000-0003-2638-9635}%
   {Work on this paper was partially supported by a NSF AF award
      CCF-1907400 and CCF-2317241.}%

   \author{Timothy Zhou}%
   {Department of Computer Science, University of Illinois, 201
      N. Goodwin Avenue, Urbana, IL 61801, USA}%
   {z.timothy96@gmail.com}%
   {https://orcid.org/0009-0009-1831-4634}%
   {Work on this paper was partially supported by a NSF AF award
      CCF-1907400.}%
}%

\CGTVer{%
   \authorrunning{S. Har-Peled and T. Zhou}%
   \Copyright{Sariel Har-Peled and Timothy Zhou}%
   \keywords{curves, congestion, packedness, \Frechet distance,
      approximation algorithms}%
   \relatedversion{\url{https://arxiv.org/abs/2105.10776}}
   \acknowledgements{%
      Work on this paper was initiated by discussions of the first
      author with Anne Driemel.  The authors thank the anonymous
      referees for their detailed and insightful comments that had
      significantly improved this work.  }
}%

\RegVer{%
   \author{%
      Sariel Har-Peled%
      \SarielThanks{Work on this paper was partially supported by a
         NSF AF award CCF-1907400.  %
      }%
      \and %
      Timothy Zhou%
      \TimThanks{}%
   }%
}

\SoCGVer{%
   \authorrunning{S. Har-Peled and T. Zhou} %
   \Copyright{Sariel Har-Peled and Timothy Zhou}%
   \ccsdesc[500]{Theory of computation~Computational geometry}%
   \keywords{curves, congestion, packedness, \Frechet distance,
      approximation algorithms} }

\RegVer{%
   \date{\today}%
}%

\begin{document}

\maketitle

\begin{abstract}
    The congestion of a curve is a measure of how much it zigzags
    around locally.  More precisely, a curve $\pi$ is $c$-packed if
    the length of the curve lying inside any ball is at most $c$ times
    the radius of the ball, and its congestion is the minimum $c$ for
    which $\pi$ is $c$-packed. This paper presents a randomized
    $42$-approximation algorithm for computing the congestion of a
    curve (or any set of segments in the plane). It runs in
    $O( n \log^2 n)$ time and succeeds with high probability.
    Although the approximation factor is large, the running time
    improves over the previous fastest constant approximation
    algorithm \cite{gsw-appc-20}, which took $\tldO(n^{4/3})$ time. We
    carefully combine new ideas with known techniques to obtain our
    new near-linear time algorithm.
\end{abstract}

\section{Introduction}

In 2010, Driemel \etal \cite{dhw-afdrc-12} provided a measurement of
how ``realistic'' a curve is (there are several alternative
definitions -- see \cite{dhw-afdrc-12} and references therein for
details). Formally, a curve $\pi$ is \emphi{$c$-packed} if the total
length of $\pi$ inside any ball is bounded by $c$ times the radius of
the ball.  The minimum $c$ for which the curve is $c$-packed is the
\emphi{congestion} of the curve. Intuitively, if a curve has high
congestion, then it zigzags back and forth around some locality. We
expect that many real-world curves do not behave so pathologically --
instead, we expect them to exhibit low congestion. For examples of
real world curves, see \url{https://frechet.xyz}. Naturally, if the
curve is a tracking of an entity over long enough time, the congestion
might be high (for example, a soccer player movement during a whole
game, an airplane path over a month, etc).

Curves with low congestion lend themselves to efficient algorithms.
Notably, they can be efficiently approximated by simpler curves which
nearly preserve \Frechet distances, and therefore the \Frechet
distances between them can be approximated in near-linear time
\cite{dhw-afdrc-12}.  In general, if the congestion is $\Omega(n)$,
then computing the \Frechet distance is more difficult. Indeed,
assuming the Strong Exponential Time Hypothesis (\SETH), even
approximating \Frechet distance within a constant factor requires
quadratic time \cite{b-wwdtt-14}.  The decision version of the problem
is also conjectured to be \ThreeSUM-hard \cite{a-cgcs}. Note that
proving a direct connection between \ThreeSUM-hardness and \SETH is
still an open problem \cite{w-hepbh-15}.

We would like to verify that a given curve is indeed $c$-packed for a
low value of $c$. Some algorithms for $c$-packed curves do not require
knowing the value of the congestion -- rather, their analyses show
that if the curve is $c$-packed for some small $c$, then the
algorithms run in near-linear time. However, verifying that curves are
$c$-packed would increase our confidence that these algorithms are
generally applicable.  This leads to the question of how quickly can
one estimate or compute the congestion. In this paper, we present a
constant-factor approximation algorithm for the congestion that runs
in near-linear time.

\myparagraph{Disks and squares are all the same.}  As far as the
congestion is concerned, whether we use disks/balls or squares/cubes
in the definition is the same up to a constant factor (i.e.,
$\sqrt{2}$ in the plane). Thus, in the following, we work usually in
settings of using squares, since it is easier.

\paragraph*{Previous work.}

As mentioned earlier, the concept of $c$-packedness was introduced in
Driemel \etal \cite{dhw-afdrc-12}.  Computing the congestion exactly
runs into the issue of minimizing sums of square roots.  As such, an
exact algorithm in the standard RAM model is unlikely.  The work of
Vigneron \cite{v-gosaf-14} provided a $(1+\eps)$-approximation
algorithm which runs in
$O\bigl( (n/\eps)^{d+2} \log^{d+2} (n/\eps) \bigr)$ time. Gudmundsson
\etal offered a cubic-time algorithm for this problem
\cite{gks-alhctd-13}. Aghamolaei \etal \cite{akgm-wqmse-20} gave a
$(2+\eps)$-approximation for the problem of approximating the
congestion of planar curves, but unfortunately their running time
seems to be at least quadratic.

Some of the previous work \cite{gks-alhctd-13} was also concerned with
computing hotspots, which are small regions in the plane where many
segments pass through. Computing regions of high congestion naturally
leads to finding hotspots. However, one usually fixes the resolution
of the desired hotspots (i.e. the radius of the balls being
intersected) before searching them.

More recently, Gudmundsson \etal \cite{gsw-appc-20} gave a
$(6+\eps)$-approximation to the congestion/packedness of a polygonal
curve in $\tldO(n^{4/3}/\eps^4)$ time. As part of their algorithm,
they showed that one can quickly find a set of $O(n)$ squares whose
congestions yield a constant approximation to the congestion of the
curve. They also noted that the problem of finding the congestion of
these squares with respect to the segments seems quite similar to the
Hopcroft problem, as discussed below.

\myparagraph{The Hopcroft problem.}

Given a set $\PS$ of $n$ points in the plane, and a set $\LSet$ of $n$
lines in the plane, the question is whether there is a point of $\PS$
that is incident to one of the lines of $\LSet$.  There is an
$\Omega(n^{4/3})$ lower bound for Hopcroft's problem due to Erickson
\cite{e-rcsgp}.  As usual, this lower bound holds only in a restricted
model of computation (algebraic decision tree model), but it is
believed to hold in broader models of computation. The belief that
this lower bound is correct (up to maybe some polylog noise) in any
reasonable model of computation is quite important as it provides
matching lower bound to the main results known in range searching.

In our language, the Hopcroft problem can be stated as having a set of
$n$ segments, and a set of $n$ (disjoint) squares, and asking for the
maximum congestion of the squares in relation to the segments.

\myparagraph{The challenge.}  If we replace a point by a short
segment, then deciding the packedness for sufficiently small squares
centered at these points is equivalent to solving the Hopcroft
problem. Naturally, this reduction of hardness only works if the
approximation is quite small (say $<2$-approx).  For this reason,
Gudmunsson \etal \cite{gsw-appc-20} deemed it unlikely that their
approach, ``or a similar approach, can lead to a considerably faster
algorithm'' for computing congestion.  Despite this lower bound
working only if the approximation constant is small, it is quite
interesting to figure out if one can break this ``lower bound'' (by
providing a worse approximation).

\paragraph*{Our result.}

Despite this difficulty, we present a randomized $O(n \log^2 n)$ time
algorithm that provides a constant approximation to the congestion of
a set of segments in $\Re^2$.  The algorithm bypasses the barrier
presented by Hopcroft's problem by observing that, when computing
congestion, the generated instances for the Hopcroft problem have high
congestion. In such scenarios, we do not need to compute the
congestion of $\Omega(n)$ disjoint squares exactly (as required by the
Hopcroft problem).

\paragraph*{Sketch of algorithm.}

The \emphw{congestion} of a square with respect to a curve $\pi$ is
the total length of $\pi$ inside it, divided by its
sidelength. Following \cite{gsw-appc-20}, we reduce the problem of
approximating the congestion of a curve to that of computing the
congestion of $O(n)$ squares. Then we build a ``few'' quadtrees whose
cells approximate these squares, so that it suffices to compute the
congestion of the quadtrees.

We can compute the congestion of a quadtree cell by finding all the
segments of the curve which intersect it, then explicitly computing
the total length of the intersections. However, naively searching for
all the segments intersecting each of the squares takes $O(n^2)$
time. To speed things up, we store each segment at some quadtree cell
with comparable length. For a given cell, its short segments are those
stored at descendants in the quadtree, and its long segments are those
stored at ancestors. To find the congestion of a cell, we compute the
sum of its long and short congestions -- the congestions with respect
to all of its long and short segments, respectively. We can quickly
compute the short congestion of a cell by summing the lengths of short
segments in the quadtree bottom-up.

Exactly computing the long congestion is somewhat tricky. Fortunately,
approximating the long congestion only requires counting the maximum
number of long segments intersecting any cell. If every cell
intersects only a few long segments, then we can quickly enumerate all
the intersections by searching the quadtree top-down. If there is a
cell intersecting many long segments, then the algorithm performs an
exponential search to guess how many intersections it has. The
algorithm quickly verifies its guess by taking a random sample of
input segments, enumerating its intersections with each quadtree cell,
and using the counts to estimate the maximum number of intersections.

\myparagraph{Highlight.}

This work ``bends'' what was previously believed to be a lower bound
on the running time for approximating congestion (i.e.,
$\Omega(n^{4/3})$). While most of the tools we use are standard, the
way we combine them is non-trivial and offers new insights into the
problem.

Since we lose constant factors in several places, the resulting
approximation factor is quite bad compared to previous work (i.e.,
$42$ vs. $6$). However, our algorithm runs in near-linear time, which
is significantly faster.  We find this result surprising, as it
bypasses the barrier formed by the Hopcroft problem mentioned above.

\paragraph*{Paper organization.}
We start in \secref{prelims} by providing some definitions and
background.  In \secref{long:and:short}, we reduce the problem of
computing the congestion to that of computing the congestions of a
small number of quadtrees. This in turn reduces to the problem of
computing the congestion from long and short segments.  The short
congestion is handled in \secref{short:alg}, while the main challenge
of approximating the long congestion is addressed in \secref{long},
where we approximate the load of the long segments -- i.e., the
maximum number of long segments intersecting a single cell of the
quadtree.  In \secref{together}, we put everything together. We
conclude in \secref{conclusions} with a few remarks.

\section{Preliminaries}
\seclab{prelims}

\subsection{Standard tools}

\begin{definition}
    For a real positive number $\num$, let $\Grid_\num$ be the
    \emphi{grid} partitioning the plane into axis-parallel squares of
    sidelength $\num$. Formally, this grid is defined by the mapping
    $\Grid_\num(x,y) = \bigl( \floor{x/\num}, \floor{y/\num} \bigr)$.
    The number $\num$ is the \emphi{width} or \emphi{sidelength} of
    $\Grid_\num$.  For integers $i,j$, the $(i,j)$-\emphi{grid cell}
    is the $\num \times \num$ square formed by the set
    $\Grid_\num^{-1}(i,j)$.
\end{definition}

\begin{definition}
    A square is a \emphi{canonical square} if it is contained inside
    the unit square, it is a cell in a grid $\Grid_w$, and $w$ is a
    power of two. That is, the square corresponds to a node in the
    infinite quadtree defined over $[0,1)^2$.  The grid generating a
    canonical square is a \emphi{canonical grid}.
\end{definition}

\subsection{Congestion}
\begin{definition}
    Let $\sq = \sqY{\pp}{r}$ denote the axis-parallel square in
    $\Re^2$ centered at a point $p \in \Re^2$ with \emphi{sidelength}
    $2r$.  The square $\sq$ can be interpreted as a ball in the
    $L_\infty$ norm, and as such, its \emphi{radius} is $r$.
\end{definition}

\begin{definition}
    For a segment $\seg$, let $\lenX{\seg}$ denote the length of
    $\seg$. Similarly, for a set of segments $\SS$, let
    $\lenX{\SS} = \sum_{\seg \in \SS} \lenX{\seg}$ denote the total
    length of segments in $\SS$.

    For a square $\sq = \sqY{\pp}{r}$, the \emphi{conflict list} of
    $\sq$ (for a set $\SS$ of segments) is the set of segments
    intersecting $\sq$, that is
    \begin{equation*}
        \clX{\sq}
        =%
        \Set{\seg \in \SS}{ \seg \cap \sq
           \neq \emptyset }.
    \end{equation*}
    Let
    \begin{equation*}
        \SS \sqcap \sq= \Set{ \seg \cap \sq}{ \seg \in \clX{\sq}
        }
    \end{equation*}
    be the clipping of the segments of $\clX{\sq}$ to $\sq$.  The
    \emphi{congestion} of the square $\sq$, with respect to $\SS$, is
    \begin{equation*}
        \congX{\sq} = \congOnY{\sq}{\SS} = \lenX{ \SS \sqcap \sq } /r.
    \end{equation*}
    The \emphi{congestion} of the set of segments $\SS$ is
    $\congX{\SS} = \underset{\pp, r}{\max} \; \congOnY{\sqY{\pp} {r}
       \bigr.}{\SS}$.  Given a set of squares $\SSet$, its
    \emphi{congestion} is
    $\congOnY{\SSet}{\SS} = \max_{\sq \in \SSet} \congX{\sq}$.
\end{definition}

\begin{definition}
    For a constant $c > 0$, a set $\SS$ of segments in $\Re^2$ is
    \emphi{$c$-packed} if, for any point $\pp \in \Re^2$ and any value
    $r > 0$, the total length of the segments of $\SS$ inside a square
    $\sq = \sqY{\pp}{r}$ is at most $c r$. That is, the congestion of
    $\SS$ is at most $c$.
\end{definition}

Thus, the congestion of $\SS$ is the minimum $c$ for which $\SS$ is
$c$-packed.  We are interested in approximating $\congX{\SS}$.  To
this end, we follow Gudmundsson \etal \cite{gsw-appc-20}, who reduced
the problem to querying the lengths of intersections between the curve
and some squares. While Gudmundsson \etal state their result for a
curve, it holds for any set of segments.

\begin{lemma}[Lemma 12 in \cite{gsw-appc-20}]
    \lemlab{wspd}%
    Given a set $\SS$ of $n$ segments in the plane, and a parameter
    $\eps \in (0,1)$, one can compute, in $O(n \log n + n / \eps^2)$
    time, a set $\GSS$ of $O(n/\eps^2)$ axis-aligned squares, such
    that
    $\congX{\SS} \geq \congOnY{\GSS}{\SS} \geq \congX{\SS}/(6+\eps)$.
\end{lemma}

\begin{remark}
    For completeness, we sketch informally an alternative proof for (a
    weaker version of) \lemref{wspd}. Let $P$ be the set of endpoints
    of the segments of $\SS$, and construct a (randomly translated)
    compressed quadtree $T$ for the points of $P$. For every node $v$
    in $T$, we add the square $\Cell_v$ to the set of candidate
    squares $\GSS$, and we also add a few scaled copies, say
    $2\Cell_v$ and $4\Cell_v$ to $\GSS$.

    Let $\sq$ be the square with maximum congestion $\congX{\SS}$.  We
    replace $\sq$ by a square $\sq' \supseteq \sq$ that is centered in
    one of the endpoints of $P$, with congestion at least
    $\congX{\SS}/2$ (the interesting case is when $\sq$ contains no
    point of $P$, but then just enlarge it till it does).

    We then continue enlarging $\sq'$ till it hits another point of
    $P$. Let $\sq''$ be the resulting square. Informally, the loss in
    congestion is a constant. Now, the center of $\sq''$ and a point
    on its boundary both belong to $P$.  With constant probability
    $\sq''$ is fully contained in a cell $\Cell_v$ of a node $v$ of
    the quadtree (or its enlarged copied added explicitly to $\GSS$),
    that is only a constant factor bigger. Thus, a square in $\GSS$
    provides the desired approximation to the congestion.
\end{remark}

\section{The algorithm: The long and short of it}
\seclab{long:and:short}

\subsection{Reduction to quadtrees}

In what follows, let $\GSS$ be the set of squares computed by
\lemref{wspd} for $\SS$ (the value of $\eps$ would be specified
shortly).  To approximate $\congX{\SS}$, it suffices to approximate
$\congOnY{\SS}{\GSS}$.

Since congestion is invariant under translation and scaling, we might
as well assume that $\SS, \GSS \subseteq [0,1/8]^2$. We can randomly
scale both sets by a random number $\rScale \in [1,2]$, and shift both
sets by a random vector $\rShift \in [0,1/2]^2$. Let
$\Lambda(p) = \rScale p + \rShift$ be the resulting affine mapping.
We compute for each square $\sq \in \Lambda(\GSS)$ the smallest
canonical square $\canonX{\sq}$ that contains it -- conceptually,
consider the infinite quadtree, and computing the lowest node $v$ in
the quadtree that its cell contains $\sq$ -- the square $\Cell_v$ is
$\canonX{\sq}$. Algorithmically, $\canonX{\sq}$ can be computed in
$O(1)$ time using the floor function and some basic bit operations,
see \cite{h-gaa-11}.  Next, we build a quadtree that has all these
marked canonical squares as nodes. This can be done in $O(n \log n)$
time \cite{h-gaa-11}. The idea is to repeat this process sufficient
number of times, such that in one of the generated quadtrees,
$\canonX{\sq}$ and $\sq$ are almost identical, and this holds for all
the squares of interest.

\begin{lemma}
    \lemlab{containment}
    For a square $\widehat{\sq} \in \GSS$, and a parameter
    $\eps \in (0,1/2)$, consider its randomly scaled and shifted copy
    $\sq = \rScale \widehat{\sq} + \rShift$, and its canonized version
    $\sq' = \canonX{\sq}$.

    The probability that
    $r = \radiusX{\sq} \leq \radiusX{\sq'} \leq (1+\eps)r$ is
    $\geq (\eps/8)^3$.
\end{lemma}
\begin{proof}
    By construction $\sq \subseteq \sq'$.  Let
    $\ell = \radiusX{\smash{\widehat{\sq}}}$, and let
    $L = 2^{\ceil{\log_2 \ell }}$ be the rounding up of $\ell$ to its
    closest power of $2$. For simplicity of exposition assume that
    $\zeta = L/l > (1+\eps)$ -- as otherwise, one can apply the
    analysis to $2L/\ell$.

    The optimal scaling for our purposes is $\zeta$, but let us be
    slightly less greedy, and consider the random scaling $\rScale$ to
    be \emphw{good} if
    $\rScale \in [(1-\eps/4)\zeta, (1-\eps/8)\zeta]$.  This interval
    is of length $(\eps/8)\zeta \geq \eps/8$. Thus, the random
    scaling is good with probability $\geq \eps/8$, and assume this
    happened.

    Consider the canonical grid $\Grid_{2L}$, and $r =\radiusX{\sq}$.
    We have that $2r \in [(1-\eps/4)2L, (1-\eps/8)2L]$. as such, the
    set of all good translations
    \begin{equation*}
        U =
        \Set{ (x,y) \in \Re^2 }{ (x,y) +
           \rScale \widehat{\sq} \text{ contained in a cell of }
           \Grid_{2L}
        }
    \end{equation*}
    is a grid like set with sidelength $2L$, where each grid point is
    replaced by a square of sidelength $\geq (\eps/8)2L$. The
    probability that the random shifting $\rShift$ falls in $U$ is at
    least $(\eps/8)^2$.

    If both things happen, then
    $r \leq L \leq (1+\eps)(1-\eps/4)L \leq (1+\eps) r$. Namely,
    $\sq \subseteq \sq'$, and
    $r = \radiusX{\sq} \leq \radiusX{\sq'} \leq (1+\eps)r $, as
    desired.
\end{proof}

\begin{definition}
    \deflab{cong:q:t}%
    For a compressed quadtree $\QT$, let $\SSet(\QT)$ be the set of
    all squares formed by nodes of $\QT$.  The \emphi{congestion} of a
    quadtree $\QT$ for a set of segments $\SS$ is the quantity
    $\congOnY{\QT}{\SS} = \congOnY{\SSet(\QT)}{\SS}$.
\end{definition}

\begin{lemma}
    \lemlab{canonical}%
    Given a set $\SS$ of $n$ segments, one can compute, in
    $O(n \log^2 n)$ time, $m = O( \log n)$ (shifted) compressed
    quadtrees, $\QT_1, \ldots, \QT_m$, such that
    \begin{math}
        \frac{1}{7} \congX{\SS}\leq%
        \max_{i} \congOnY{\QT_i}{\SS} \leq \congX{\SS}.
    \end{math}
    This holds with high probability.
\end{lemma}

\begin{proof}
    Let $\eps \in (0,1)$ be a sufficiently small constant to be
    specified shortly. Compute the set $\GSS$ of squares specified by
    \lemref{wspd}. Next, compute $m = O( \eps^{-3} \log n )$ randomly
    shifted copies $\GSS^1, \ldots, \GSS^m$ of $\GSS$, as described
    above. For each one of them we compute a compressed quadtree. Each
    such compressed quadtree, can be interpreted as a shifted
    compressed quadtree over the original set of squares $\GSS$. And
    let $\QT_1, \ldots, \QT_m$ be these quadtrees.

    Fix a square $\sq \in \GSS$. By \lemref{containment}, with
    probability $p \geq (\eps/8)^3$, there is a node $u$, and thus its
    cell $\sq'$, such that $\sq \subseteq \sq'$ and
    $\radiusX{\sq'} \leq (1+\eps)\radiusX{\sq}$.  Thus, we have
    $\congOnY{\sq'}{\SS} \geq \congOnY{\sq}{\SS}/(1+\eps)$.

    In particular, the probability that this ``good'' containment does
    not happen for $\sq$, in all $m$ quadtrees, is
    $(1-p)^m \leq 1/n^{O(1)}$. We conclude that, high probability, for
    all squares $\sq\in\GSS$, there is at least one node/square, in
    the computed quadtrees, that tightly contains $\sq$.

    This readily implies that
    \begin{equation*}
        \max_{i} \congOnY{\QT_i}{\SS}
        \geq
        \frac{1}{1+\eps}\max_{\sq \in \GSS} \congOnY{\sq}{\SS}
        \geq
        \frac{1}{(1+\eps)(6+\eps)} \congX{\SS}
        \geq
        \frac{1}{7} \congX{\SS},
    \end{equation*}
    for $\eps = 1/10$.
\end{proof}

\begin{remark}
    Throughout the algorithm, we use compressed quadtrees, rather than
    regular quadtrees. The use of compressed quadtrees is necessary to
    get an efficient runtime, but it does not affect the description
    of our algorithm significantly. (Intuitively, the
    compressed quadtree compresses paths and not cells, so one can easily
    ensure that all cells of interest appear in the compressed
    quadtree as regular cells). As such, from this point on, we use
    quadtree as a shorthand for a compressed quadtree, and we ignore
    the minor low-level technical details that arise because of the
    compression.

    Specifically, the only issue in the algorithm where this would
    make a difference is in traversing down the tree and assuming that
    the child is exactly half the size of its parent. This can easily
    be fixed by working directly with the node own record of its
    dimensions.
\end{remark}

\subsection{First steps towards approximating %
   the congestion of a quadtree}

Given a set $\SS$ of $n$ segments and a quadtree $\QT$ of size $O(n)$,
we wish to compute the congestion of $\QT$.  To simplify the
exposition, we assume that any shift needed is applied to the input
segments, and the quadtree is thus standard one built over $[0,1]^2$.

\subsubsection{A \naive exact algorithm for %
   the congestion of a quadtree}
\seclab{naive:1}

At each node $v \in \QT$, the algorithm stores a \emphi{conflict list}
$\clX{v}$ -- a list of all the segments of $\SS$ intersecting
the cell $\sq_v$ associated with $v$. It begins by storing $\SS$ at
the root of $\QT$, then recursively traverses down the tree.  At
each parent node $u$, the algorithm sends the list $\clX{u}$ to all
its children. Each child $v$ finds all the segments which intersect
its cell $\sq_v$ and adds them to its own list.  At
the end of this process, the algorithm has computed the segments
intersecting each quadtree node, and it can use them to compute the
congestion. Overall, this algorithm takes $O(n^2)$ time.

\subsubsection{The long and the short of it}

To speed up the \naive algorithm, we implement several strategies. The
first is to register the segments directly in the cells at a suitable
resolution of the quadtree.

\begin{definition}[long/short threshold]
    Let $\tr > 0$ be a fixed integer constant, to be specified
    shortly\footnote{Spoiler alert! The butler did it, $\tr=20$, and
       the hero dies in the end. The hero did try to expose $\tr$, but
       it was too late for them.}. The parameter $\tr$ is the
    \emphi{transition} constant.  A segment $\seg$ is
    \emphi{$\tr$-long} (resp., \emphi{$\tr$-short}) for a square
    $\sq = \sqY{\pp}{r}$ if it intersects the interior of $\sq$ and
    $\lenX{\seg} \geq \tr r$ (resp., $\lenX{\seg} < \tr r$). For a
    square $\sq$ with radius $r$, its set of $\tr$-long (resp.,
    $\tr$-short) segments is denoted by $\longSY{\tr}{\sq}$ (resp.,
    $\shortSY{\tr}{\sq}$).
\end{definition}

\begin{figure}
    \phantom{} \hfill%
    \includegraphics[page=1]{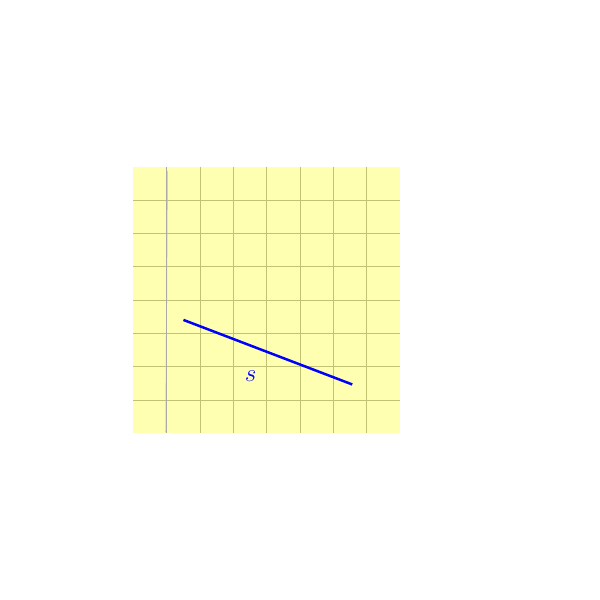}%
    \hfill%
    \includegraphics[page=2]{figs/grid_seg}%
    \hfill%
    \phantom{}%
    \caption{Computing a maximal set of $\tr$-long segments for a
       segment $\seg$ (see \lemref{register}). In this case, the
       segment has length $5$, and each grid cell is of radius $0.5$,
       so the segment is $10$-long for all the cells of the grid it
       intersects.}
    \figlab{canonical}
\end{figure}

\begin{lemma}
    \lemlab{register}%
    For each segment $\seg \in \SS$, let $\GridSetY{\seg}{\tr}$ be the
    set of interior-disjoint canonical squares of maximal size for
    which $\seg$ is $\tr$-long. There are at most $O(1 + \tr)$ such
    squares, and they can be computed in $O(1 + \tr)$ time.
\end{lemma}
\begin{proof}
    Let $i = \floor{ \log_2 (\lenX{\seg}/\tr)}$, and $\num=2^i$. Since
    \begin{equation*}
        \tr \num
        =
        \tr 2^{\floor{\log_2(\lenX{\seg}/\tr)}}
        \leq%
        \tr \cdot \frac{\lenX{\seg}}{\tr}
        \leq%
        \lenX{\seg}
        < \tr 2^{1+\floor{\log_2(\lenX{\seg}/\tr)}}
        <%
        2 \tr \num,
    \end{equation*}
    the segment $\seg$ intersects at most $2(\tr+1)$ horizontal lines
    of the grid $\Grid_\num$. This implies that $\seg$ can intersect
    at most $4(\tr+1) +1$ grid cells of $\Grid_\num$. Computing the
    grid cells that intersect $\seg$ is a classical problem in
    graphics (i.e.,
    \href{https://en.wikipedia.org/wiki/Line_drawing_algorithm}{line
       drawing problem}) which can be solved in $O(1+\tr)$ time; see
    \figref{canonical}.
\end{proof}

\subsubsection{Registering the segments}
\seclab{register}

\myparagraph{The refined quadtree $\QTM$.}
\seclab{refine}

Given the set $\SS$ of $n$ segments and the above quadtree $\QT$, the
algorithm first computes the set of canonical squares
$\SSet = \mathrm{cells}(\QT) \cup \bigcup_{\seg \in \SS}
\GridSetY{\seg}{\tr}$; see \lemref{register}. The algorithm then
computes the (compressed) quadtree $\QTM$ for $\SSet$. This can be
done in $O( n \log n)$ time \cite{h-gaa-11}, as $\cardin{\SS} = n$,
$\cardin{\mathrm{cells}(\QT)} = O(n)$, and
$\cardin{\SSet} = O( ( 1+ \tr ) n ) = O(n)$, as $\tr$ is a constant.
Each square of $\SSet$ is now present as a cell of a node of the
computed quadtree.

The algorithm now stores every segment $\seg \in \SS$ in the cells of
$\GridSetY{\seg}{\tr}$. Each cell $\sq \in \GridSetY{\seg}{\tr}$ corresponds
to a node $v$ in the quadtree, and the algorithm stores $\seg$ in
$\clX{v}$.  This takes $O(n \log n)$ time; see
\cite{h-gaa-11}. Thus, for every quadtree node $v$, the
algorithm computes a list $\clLX{v}$ of segments registered there.
Segments registered in this list are long for the cell $\sq_v$ but short for
cells in higher levels of the quadtree.

Propagating each segment up to the parent node, we also register each
segment in a list $\clSX{v}$ of another node $v$ (this is done only
for one level up). Segments registered in this list are short for the
cell $\sq_v$ but long for cells in lower levels of the quadtree. Since
computing the registering cells for each segment takes $O(1)$ time,
computing the lists $\clSX{\cdot}$ and $\clLX{\cdot}$ for all nodes in
the tree takes $O(n)$ time, and the lists themselves have total length
$O(n)$.

A segment is registered only once as a short or long segment on any path in the quadtree.

\begin{definition}
    \deflab{long:short}%
    The \emphi{$\tr$-long congestion} of $\sq$ is
    $\longcY{\tr}{\sq} = \lenX{ \longSY{\tr}{\sq} \sqcap \sq
    }/r$. Similarly the \emphi{$\tr$-short congestion} of $\sq$ is
    $\shortcY{\tr}{\sq} = ||\shortSY{\tr}{\sq} \sqcap \sq ||/r$.
    Given a quadtree $\QT$, its \emphi{$\tr$-long congestion} and
    \emphi{$\tr$-short congestion} are
    \begin{equation*}
        \longc
        =
        \longcY{\tr}{\QT} = \max_{\sq \in \mathrm{cells}(\QT)}
        \longcY{\tr}{\sq}
        \qquad\text{and}\qquad%
        \shortc
        =
        \shortcY{\tr}{\QT} = \max_{\sq \in \mathrm{cells}(\QT)}
        \shortcY{\tr}{\sq}.
    \end{equation*}
\end{definition}

For a node $v \in \QTM$, let $\ancX{v}$ (resp., $\descX{v}$) be the
list of ancestors (resp., descendants) of $v$ in the tree $\QTM$. Here
(emptily) $v \in \ancX{v}$ and $v \in \descX{v}$.  Consider a node
$v \in \QTM$, and let $\sq_v$ be its associated square. We have that
\begin{equation*}
    \longSY{\tr}{\sq_v}%
    =%
    \bigcup_{u \in \ancX{v}} \pth{\clLX{u} \cap \sq_v}
    \qquad\text{and}\qquad%
    \shortSY{\tr}{\sq_v}%
    =%
    \bigcup_{u \in \descX{v}} \clSX{u}.
\end{equation*}

To summarize, we have described how to augment a quadtree $\QT$
by adding more cells and registering short and long segments at the cells,
yielding a new quadtree $\QTM$. In what follows, we
use this stored information to compute the long and short congestions
of $\QTM$ (which are at least the congestions of the
sub-quadtree $\QT$).

\subsection{Computing the congestion of the short segments}
\seclab{short:alg}

The $\tr$-short congestion is computed using dynamic programming, as
described next.

\begin{lemma}
    \lemlab{short:alg}%
    Given a set $\SS$ of $n$ segments in the plane and a quadtree
    $\QTM$ of size $O(n)$, one can compute, in $O(n )$ time, the
    $\tr$-short congestion for all the nodes of $\QTM$, where $\tr$ is
    a constant.
\end{lemma}
\begin{proof}
    We compute the $\tr$-short congestion of a quadtree via dynamic
    programming. The algorithm finds the total length of short
    segments intersecting each leaf and propagates the values upward.

    For every node $v \in \QTM$, the algorithm computes the quantity
    $\lenX{\clSX{v} \sqcap \sq_v}$. Computing the value for node $v$
    requires time proportional to the total size of the list
    $\clSX{v}$, so doing it for all the nodes of the tree takes $O(n)$
    time overall.  Next, the algorithm traverses the tree bottom-up.
    For each node along the way, it computes the total lengths of the
    intersecting short segments:
    \begin{equation*}
        \lenTShortX{v}%
        =%
        \sum_{u \in \descX{v}} \lenX{\clSX{u} \sqcap \sq_u}
        =%
        \lenX{\clSX{v} \sqcap \sq_v} + \sum_{u \text{ child of }v}
        \lenTShortX{u}.
    \end{equation*}
    It is easy to verify that
    $\lenTShortX{v} = \lenX{\shortSY{\tr}{\sq_v} \sqcap \sq_v}$. The
    lemma follows.
\end{proof}

\section{Approximating the maximum load of the long segments %
   in a quadtree}
\seclab{long}

Handling the long congestion quickly seems challenging. Instead, here
we would quickly approximate a proxy for this quantity -- the maximum
number of such segments visiting a single node.

\begin{definition}
    \deflab{max:sz}%
    The \emphi{$\tr$-load} of a quadtree $\QTM$ is the quantity
    \begin{math}
        \maxSz = \max_{\sq \in \QTM} |\longSY{\tr}{\sq}|.
    \end{math}
\end{definition}

\subsection{A \naive algorithm for computing the $\tr$-load}

To compute its exact load of a given square, we need to find the long
segments intersecting the cell. We can compute the conflict list for
each cell by pushing long segments downward from the conflict lists of
its ancestors (as done in the \naive algorithm of \secref{naive:1}).
Unfortunately, these lists can get quite long.

\begin{lemma}
    \lemlab{silly}%
    One can compute the $\tr$-load congestion of $\QTM$ in
    $O(n \log n + \maxSz n )$ time.  More generally, given a set
    $\Sample \subseteq \SS$ and a threshold $t$, one can decide
    whether
    $\maxSzX{\Sample} = \max_{\sq \in \QTM} |\Sample \cap
    \longSY{\tr}{\sq}| \leq t$, in $O(n \log n + t n)$ time.
\end{lemma}
\begin{proof}
    The algorithm starts with the precomputed lists $\clLX{u}$ and
    traverses the tree top-down. At each node, it pushes the stored
    list down to the children. Each child $v$ selects the segments of
    the incoming list that intersect its cell, and takes the union of
    this filtered incoming list with $\clLX{v}$. This yields the
    conflict list of $v$ of \emph{all} $\tr$-long segments that
    intersect it, denoted by $\clALX{v}$. It then pushes this list
    down to its children, and so on.

    For each node $v$ of the quadtree, it is now straightforward to
    compute the congestion of the segments of $\clALX{v}$
    (or the length of the list $\clALX{v}$)  for the cell of the node $v$.  It
    follows that this computes the $\tr$-long congestion for each node
    of the quadtree.  Since the maximum length of the lists sent down
    is $\maxSz$, the claim follows.

    The second algorithm works in a similar fashion, except that the
    algorithm propagates downwards only segments of $\Sample$. If in
    any point in time the computed conflict list gets bigger than $t$,
    the algorithm bails out, returning that $\maxSzX{\Sample} > t$.
\end{proof}

\begin{figure}[b]
    \centerline{%
       \includegraphics{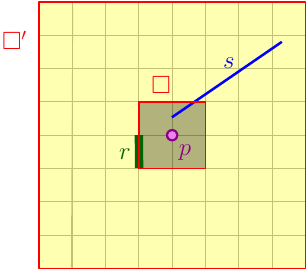}%
    }
    \caption{An $\tr$-long segment for a square
       $\sq = \sqY{\pp}{r}$ intersects the square
       $\sq' = \sqY{\pp}{(1+\tr)r}$ with a segment of length at
       least $\tr r$ (here, $\tr=3$).}
    \figlab{long}
\end{figure}

\subsection{From long congestion to load, and back}

Instead of computing the $\tr$-long congestion of a quadtree exactly,
we content ourselves with a constant-factor approximation, by using
the load instead.

\begin{lemma}
    \lemlab{long:cong}%
    For a cell $\sq = \sqY{\pp}{r}$ of $\QTM$, we have
    \begin{math}
        \displaystyle%
        \frac{\longcY{\tr}{\sq}}{\sqrt{8}}%
        \leq%
        |\longSY{\tr}{\sq}|%
        \leq%
        \frac{1+\tr}{\tr} \congX{\SS}.
    \end{math}
\end{lemma}

\begin{proof}
    The intersection of each segment of $\longSY{\tr}{\sq}$ with $\sq$
    can have length at most $\sqrt{2}\cdot 2 r = \sqrt{8}r$ (as the
    sidelength of $\sq$ is $2r$), so
    $\longcY{\tr}{\sq} = \lenX{ \sq \sqcap \longSY{\tr}{\sq}}/ r$, and
    $\lenX{ \sq \sqcap \longSY{\tr}{\sq}}\leq |\longSY{\tr}{\sq}|
    \cdot r \sqrt{8}$, and the first inequality follows.

    As for the second inequality, consider
    $\sq' = \sqY{\pp}{(1+\tr)r}$. By definition, each long segment
    $\seg \in \longSY{\tr}{\sq}$ has length at least $\tr r$ and
    intersects $\sq$. The length of $\seg \cap \sq'$ is at least $\tr
    r$, see \figref{long}.
    As such, we have
    \begin{equation*}
        \cardin{\longSY{\tr}{\sq}} \cdot \frac{\tr r}{(1+\tr)r} %
        \leq%
        \frac{\lenX{\longSY{\tr}{\sq} \sqcap \sq'}}{(1+\tr)r}%
        \leq%
        \congX{\sq'}
        \leq%
        \congX{\SS}.
        \CGTVer{\qedhere}%
    \end{equation*}
\end{proof}

\subsection{Exponential search for maximum $\tr$-load}

It remains to estimate the load (i.e., maximum number of $\tr$-long
segments intersecting any cell in the quadtree $\QTM$).  The basic
idea is to couple random sampling together with exponential search, to
approximate the maximum load. If during a round an ``overflow''
occurs, the algorithm increases the guess for the load by (say a
factor of $2$), and moves on to the next round.  Using sampling and
exponential search to estimate a quantity is an old idea. In the
context of geometric settings, it was used before to estimate the
maximum depth of nicely behaved regions, see \cite{ah-adrp-08}.

\subsubsection{The algorithm}

Our purpose here is to approximate the $\tr$-load of $\QTM$, denoted
by $\maxSz$, see \defref{max:sz}.

\myparagraph{Estimation via sampling.}

We perform an exponential search for the size of the largest conflict
list. We use the following standard sampling lemma. It follows by a
standard application of Chernoff's inequality -- see \cite[Lemma
2.7]{bhrrs-eeiso-20} for a proof.

\begin{lemma}  %
    \lemlab{estimate}%
    Consider two (finite) sets $\SC \subseteq \SA$, where
    $n = \cardin{\SA}$.  Let $\epsA\in (0,1)$ and
    $\BadProb \in (0,1/2)$ be parameters, and let
    $r = \ceil{\smash{\cE \epsA^{-2} \tfrac{n}{g} \log
          \BadProb^{-1}}\bigr.}$, where $\cE$ is a sufficiently large
    constant.

    Let $g > 0$ be a user-provided guess for the size of
    $\cardin{\SC}$. Consider a random sample $\Sample$, taken from
    $\SA$ by picking each element with probability $\tfrac{r}{n}$,
    Next, consider the estimate
    $Y = \tfrac{n}{r} \cardin{ \Sample \cap \SC}$ to
    $\cardin{\SC}$. Then, we have the following: \smallskip%
    \begin{compactenumA}
        \item If $Y < g/2$, then $\cardin{\SC} < g$,
        \item If $Y \geq g/2$, then
        $(1-\epsA)Y \leq \cardin{\SC} \leq (1+\epsA)Y$.
    \end{compactenumA}
    \smallskip%
    Both statements above hold with probability $\geq 1-\BadProb$.
\end{lemma}

\myparagraph{The algorithm.}

The number of cells in $\QTM$ is $O(n)$, and let
$\BadProb =1/n^{O(1)}$. Let $\eps \in (0,1)$ be a prespecified
approximation parameter, and let $\epsA = \eps/3$.  At each round, the
algorithm has a threshold number $g_i$ and checks whether (i)
$\maxSz > 8g_i$ (roughly), or (ii) $\maxSz =\Theta(g_i)$ and it can
approximated reliably and quickly.  The $i$\th round start, with
parameters
\begin{equation}
    g_i = 2^{i-1}
    \qquad\text{and}\qquad%
    r_i =
    \ceil{\smash{\cE \epsA^{-2} \tfrac{n}{g_i} \log \BadProb^{-1}}}.
    \eqlab{prob}
\end{equation}
where $\cE$ is a sufficiently large constant.  If the $i$\th round
failed, the algorithm goes on to the next round (i.e., by increasing
$i$ by one).

Let $I$ be the first integer such that $r_i < n$ (i.e.,
$I = \Theta( \delta^{-1} \log n )$), and the algorithm starts with
$i=I$.  This first round, the algorithm uses \lemref{silly}
directly. If it finds some conflict list (for $\SS$) that contains
more than $8g^{}_i = 8 g_I^{}$ segments at any point during the
execution, it concludes that $g_i$ is too small. Thus this round is be
a failure, and it goes on to the next round.  Otherwise, the algorithm
computes $\maxSz$ exactly and returns its value.

For later rounds, and larger values of $g_i$, the round begins by
taking a random sample $\Sample_i \subseteq \SS$ of the segments. Each
segment is included in $\Sample_i$ independently with probability
$\prob_i= r_i / n$ (i.e., a random sample according to
\lemref{estimate}).  The algorithm then calls the subroutine of
\lemref{silly} with $\Sample_i$ as the list of segments and with
$U_i = 8 g_i \prob_i$ as the threshold.  If the output reveals that
$\maxSzX{\Sample } > U_i$, then the guess $g_i$ is too small, the
round failed, and the algorithm continues to the next round.

If a round succeeds (and $i \neq I$), the algorithm outputs
$(1-\epsA)Y$ for $Y= \tfrac{n}{r_i}\maxSzX{\Sample }$ as an estimate
from below to $\maxSz$.

\subsubsection{Analysis}

\myparagraph{Running time.}  The first round takes $O( n\log^2 n )$
time.  Since there are only $n$ segments in $\SS$, the algorithm
always stops by round $O( \log n)$.  The threshold used in each round
is
\begin{equation*}
    U_i
    =%
    8g_i \prob_i
    =
    \frac{8g_i r_i}{n}
    =
    O\pth{  \smash{\frac{g_i n}{n \epsA^2 g_i n} } \log n^{O(1)} }
    =
    O\Bigl( \frac{\log n}{\eps^{2}} \Bigr).
\end{equation*}
As such, the running time of each round is $O(\eps^{-2} n \log n)$;
see \lemref{silly}, and the overall running time is
$O(\eps^{-2} n \log^2 n)$.

\paragraph*{The result.}

\begin{lemma}
    \lemlab{long:alg}%
    Using the above randomized algorithm, one can compute, in
    $O( \eps^{-2}n{\log^2 n} )$ time, a number $\Delta$, such that
    with high probability, we have
    \begin{math}
        \Delta \leq \maxSz \leq (1+\eps)\Delta,
    \end{math}
    where $\maxSz$ is the $\tr$-load of $\QTM$, see \defref{max:sz}.
\end{lemma}
\begin{proof}
    If the algorithm terminated after the first round, then it output
    $\maxSz$ exactly, and we are done.

    Otherwise, if it failed in a round, it must be that it found a
    node that its conflict list exceeds the ``expected'' threshold
    $U_i = 8 g_i \prob_i$. \lemref{estimate} implies that the random
    sample in such cases estimates the size of the original conflict
    list correctly, and it is at least of size $(1-\epsA)U_i > 4
    g_i$. Namely, $g_i$ is (way) too small, and this decision was made
    correctly with probability $1 - 1/n^{O(1)}$.

    If the algorithm succeeded in a round, then \lemref{estimate}
    implies that $\maxSz \leq (1+\epsA)U_i < 16 g_i$. Since the
    algorithm failed the previous round, we have that
    $\maxSz > 4g_{i-1} \geq 2g_i$.  \lemref{estimate} then implies
    that $\maxSz$ is $(1\pm\epsA)$-estimated correctly. Specifically,
    for the $j$\th ``heavy'' node (i.e., conflict list size $\maxSz_j$
    is larger than $g_i$) in the tree, let $Y_j$ be the estimate of
    its conflict list size. We have that
    $(1-\epsA)Y_j \leq \maxSz_j \leq (1+\epsA)Y_j$. This inequality
    clearly holds also on the max value, which implies that
    \begin{equation*}
        (1-\epsA)Y \leq \maxSz \leq (1+\epsA)Y
        \implies%
        Y
        \leq \maxSz \leq \frac{1+\epsA}{1-\epsA}Y
        \leq (1+ 3\epsA)Y
        =
        (1+\eps)Y.
    \end{equation*}

    An important technicality here is that any of the conflict lists
    in the tree of size smaller than $g_i/2$ are too small after the
    sampling to compete with the conflict list realizing $\maxSz$.
    Similarly, this also holds for conflict lists of size $\leq
    g_i$. Thus, all the conflict lists in play for realizing the
    maximum of the sample are estimated correctly\footnote{So,
       confusingly, while the winning estimate might not be from the
       node with the largest conflict list, it comes from a node with
       a conflict list of size very close to it -- if it was much
       smaller, it would not have been able to beat it.}.
\end{proof}

\section{Approximating the maximum congestion}
\seclab{together}

We seem to have lost our way, so lets try to get back on
track. Consider the following quantity
\begin{equation*}
    \Longc(\QTM)
    =
    \max_{\sq \in \mathrm{cells}(\QTM)}
    \max\bigl[ \longcY{\tr}{\sq},
    \congOnY{(1+\tr)\sq  \bigr.} {\longSY{\tr}{\sq}}\bigr],
\end{equation*}
which is the \emphi{augmented} long congestion of $\QTM$. We
define\footnote{Spoiler alert! The hero is still dead, and it turns
   out that for $\QT_1, \ldots, \QT_m$ the quadtrees of
   \lemref{canonical}, $\max_i \Longc(\QT_i)$ is a good approximation
   to $\congX{\SS}$.  }
\begin{equation*}
    \CongOnY{\QTM}{\SS}
    =
    \max\pth{  \shortc(\QTM),
       \Longc(\QTM)}.
\end{equation*}

Clearly,
$\congOnY{\QTM}{\SS}/2 \leq \CongOnY{\QTM}{\SS} \leq \congX{\SS}$.

\begin{lemma}
    \lemlab{silly:4}%
    We have
    $ \frac{\tr}{1+\tr} \maxSz \leq \Longc(\QTM) \leq \sqrt{8} \maxSz
    $.
\end{lemma}
\begin{proof}
    The upper bound is immediate.  The lower bound readily follows
    from the argument used in the proof of \lemref{long:cong}.
\end{proof}

\begin{lemma}
    \lemlab{a:cong:q:t}%
    One can compute, in $O(n \log^2 n)$ time, a quantity
    \begin{math}
        \ACongOnY{\QTM}{\SS},
    \end{math}
    such that
    \begin{math}
        \ACongOnY{\QTM}{\SS} \leq \CongOnY{\QTM}{\SS} \leq 3
        \ACongOnY{\QTM}{\SS}.
    \end{math}
    This holds with high probability.
\end{lemma}

\begin{proof}
    let $\eps = 1/100$, and let $\Delta$ be the approximation of
    $\maxSz$ computed by the algorithm of \lemref{long:alg}. This
    takes $O(n \log^2 n )$ time. We now compute $\shortc(\QTM)$ in
    $O(n \log n)$ time using the algorithm of \lemref{short:alg}. We
    compute the quantity
    \begin{equation*}
        \ACongOnY{\QTM}{\SS}
        =
        \max\Bigl(  \shortc(\QTM),
        \frac{\tr}{1+\tr}\Delta
        \Bigr).
    \end{equation*}
    Clearly, $\ACongOnY{\QTM}{\SS} \leq \CongOnY{\QTM}{\SS}$.
    As for the other direction, observe that
    \begin{equation*}
        \frac{\sqrt{8}(1+\eps)(1+\tr)}{\tr}
        \ACongOnY{\QTM}{\SS}
        \geq
        \frac{\tr}{1+\tr}\Delta
        \cdot
        \frac{\sqrt{8}(1+\eps)(1+\tr)}{\tr}
        \geq
        \sqrt{8}
        \maxSz
        \geq
        \Longc(\QTM),
    \end{equation*}
    by \lemref{silly:4}. Numerical calculations shows that for
    $\eps =1/100$ and $\tr=20$,
    $\frac{\sqrt{8}(1+\eps)(1+\tr)}{\tr} \leq 3$, which establish the
    claim.
\end{proof}

Finally, we arrive at our main result.
\begin{theorem}
    \thmlab{bound:final}%
    Let $\SS$ be a set of $n$ segments in the plane.  One can compute,
    in $O(n \log^3 n)$ time, a $42$-approximation to
    $\congX{\SS}$. The algorithm is randomized and succeeds with high
    probability.
\end{theorem}
\begin{proof}
    Compute, in $O(n \log^2 n)$ time, the $m=O(\log n)$ quadtrees
    $\QT_1, \ldots, \QT_m$ of \lemref{canonical}. Each quadtree
    $\QT_i$ is refined as described in \secref{refine}, into a
    quadtree $\QT_i^+$, and we compute for each one of them the
    quantity $\zeta_i = \ACongOnY{\QT^+_i}{\SS} $, using the
    algorithm of \lemref{a:cong:q:t} in $O(n \log^2 n)$ time. Overall,
    the running time is $O( n\log^3 n )$, and all the steps so far
    succeeded with high probability.  The algorithm returns
    $\zeta = \max_i \zeta_i \leq \congX{\SS}$ as the desired
    approximation to the congestion. For the quality of approximation,
    observe that
    \begin{equation*}
        42 \zeta
        =
        14 \max_i 3 \zeta_i
        \geq
        14 \max_i \CongOnY{\QT^+_i}{\SS}
        \geq
        7 \max_i \congOnY{\QT^+_i}{\SS}
        \geq
        7 \cdot
        \frac{1}{7} \congX{\SS}
        =%
        \congX{\SS}.
        \CGTVer{\qedhere}%
    \end{equation*}
\end{proof}

\section{Conclusions}
\seclab{conclusions}

We provided a near-linear time algorithm that computes a constant
factor approximation to the congestion of a polygonal curve (i.e., the
minimum $c$ such that the curve is $c$-packed). We consider the result
to be quite surprising, even though the constant is undesirably large
(i.e., $42$).

The new algorithm works verbatim in any constant dimension -- our
algorithm has not used planarity in any way.  The quality of
approximation deteriorates with the dimension $d$, but it is still a
constant when $d$ is a constant. The running time remains
$O(n \log^3 n)$.

Another important property of the new algorithm is that it does not
require the input segments to form a curve. It is natural to
conjecture that in the plane, if the input is a polygon curve that
does not self intersect, then one should be able to
$(1+\eps)$-approximate the congestion in near linear time. We leave
this as an open problem for further research.

As mentioned above, the constant in the approximation quality of the
new algorithm is not pretty (currently $42$). Reducing the constant
further while keeping the running time near-linear is an interesting
problem for future research.

Nothing in our algorithm is a no-starter from a practical point of
view. Performing an experimental study using this algorithm is an
interesting future research.

\RegVer{%
   \myparagraph{Acknowledgements.}

Work on this paper started more than a decade ago in a collaboration
between the first author and Anne Driemel.  The authors thank the
anonymous referees for their detailed and insightful comments that had
significantly improved this work.

}

\BibTexMode{%
   \CGTVer{%
      \bibliographystyle{plainurl}
   }
   \SoCGVer{%
      \bibliographystyle{plain}%
   }%
   \RegVer{%
      \bibliographystyle{alpha}%
   }%
   \bibliography{packedness} }

\BibLatexMode{\printbibliography}

\end{document}